\documentclass{svjourMy}                     
\smartqed  

\usepackage{graphicx}
\usepackage{float}
\usepackage[cmex10]{amsmath}
\usepackage{natbib}
\usepackage{caption}
\usepackage{epstopdf}
\usepackage{placeins}
\usepackage{appendix}
\usepackage{multirow}

\begin{document}

\title{A new estimate of mutual information based measure of dependence between two variables: properties and fast implementation
}
\subtitle{}

\titlerunning{Mutual Information based Dependence Index}        

\author{Namita Jain         \and
        C.A. Murthy 
}

\institute{Namita Jain \at
              Machine Intelligence Unit, Indian Statistical Institute, 203 Barrackpore Trunk Road,Kolkata 700108,India \\
              Tel.: +919874446800\\
              Fax: +91332573357\\
              \email{namita.saket@gmail.com}           
           \and
           C.A. Murthy \at
              Machine Intelligence Unit, Indian Statistical Institute, 203 Barrackpore Trunk Road,Kolkata 700108,India
              \email{murthy@isical.ac.in}
}



\maketitle

\begin{abstract}
This article proposes a new method to estimate an existing mutual information based dependence measure using histogram density estimates. Finding a suitable bin length for histogram is an open problem. We propose a new way of computing the bin length for histogram using a function of maximum separation between points. The chosen bin length leads to consistent density estimates for histogram method. The values of density thus obtained are used to calculate an estimate of an existing dependence measure. The proposed estimate is named as MIDI (Mutual Information based Dependence Index). Some important properties of MIDI have also been stated.

The performance of the proposed method has been compared to generally accepted measures like distance correlation (dcor), Maximal Information Coefficient (MINE) in terms of accuracy and computational complexity with the help of several artificial data sets with different amounts of noise. The proposed method is able to detect many types of relationships between variables, without making any assumption about the functional form of the relationship. The power statistics of proposed method illustrate their effectiveness in detecting non linear relationship. Thus, it is able to achieve generality without a high rate of false positive cases. MIDI is found to work better on a real life data set than competing methods. The proposed method is found to overcome some of the limitations which occur with dcor and MINE. Computationally, MIDI is found to be better than dcor and MINE, in terms of time and memory, making it suitable for large data sets.
\end{abstract}.

\section{Introduction}

While modelling complex systems, it is often found that non-linear dependence occurs between variables. Sometimes these relationships are functional. At other times, the relationship cannot be modelled well using functional forms. For example, in biological systems and weather forecasting the behaviour of variables may seem random, but repetition of a fixed random pattern may reveal a relationship between the variables as discussed by \cite{dalgleish}. Even if the relationship can be modelled using simple functional form, it is infeasible to test against each functional form individually as there exist infinitely many functional forms.

\subsection{Some methods commonly used to measure dependence between variables}

There are several measures of dependence between two variables in literature. The most widely used dependence measure is Correlation coefficient ($\rho$). It has several desirable properties. However, it accounts only for linear relationship between the variables. Having $\rho=0$ does not imply the independence of variables. It is necessary to have a measure which depicts non-linear relationship between two variables. Additionally, the measure should not assume any theoretical probability distribution.

The literature on dependence measures went along two different lines. Some authors concentrated on the properties that a measure should possess. \cite{Urbach:2000} and \cite{Dionisio} support a strong relationship between entropy, dependence and predictability. This relation has been studied by several authors, namely \cite{darbelly}, \cite{CGMR}, \cite{physicaA:dionisio}, \cite{reshef}, \cite{Kraskov:2008}, and \cite{Yao}. Dependence measure MINE defined by \cite{reshef} is found to be  capable of calculating dependence between variables related to each other in different ways. However, \cite{Tibshirani} state that MINE overestimates the dependence between variables.

Other authors provided different dependence measures. \cite{SRpaper} have defined the distance correlation as a dependence measure which has several desirable properties but it does not detect relationships like sinusoidal wave, circle etc. where one way relationship exists. \cite{Kvalseth} has briefly defined an information-theoretic measure of dependence between two variables in his paper. This measure has several desirable properties. However, he has not discussed how this measure can be estimated for a given set of points.
More discussion on existing mutual information based measures and their estimation can be found in section \ref{MIsec}.

In this paper we propose a method for estimating the measure $KM1$ described by Kv{\aa}lseth. For estimation of $KM1$ we need to estimate mutual information and entropy. This estimation is done using density estimates obtained from histogram. This article addresses the important issue of finding the bin size which results in consistent density estimates. The proposed method of finding bin length is efficient as it depends on simple statistical properties of data. The estimate of $KM1$ thus obtained is called Mutual Information based Dependence Index(MIDI). This index is capable of finding dependence between variables related to each other in different ways without overestimating the dependence.

The rest of the paper is organized as follows. Section \ref{MIsec} discusses some methods for estimating mutual information based measures. In this section we also look at some mutual information based dependence measures existing in literature. Section \ref{proposedSec} introduces proposed method of finding the estimate called MIDI. Section \ref{chosenMeas} explains why a particular measure has been chosen for estimation. Section \ref{BinReq} discusses the importance of choosing bin length in histogram method. Section \ref{binlength} consists of the method used to determine the bin length. Section \ref{MIChist} provides the method of calculating MIDI from the histogram, and section \ref{proposedL1} discusses the $L_1$ consistency of proposed method of density estimation. Section \ref{Algo} provides the algorithm to calculate MIDI. Section \ref{results} contains the calculated values of MIDI and other algorithms for artificial datasets with and without noise. It also gives comparison of proposed and other methods in terms of power and experiments on a real life data set. The article concludes with section \ref{conc}. An appendix has been included which contains theorems stated by other authors. All the tables and figures referred to in the text have been provided at the end of the paper.

\section{Literature survey and Mutual information based dependence measures}
\label{MIsec}

\subsection{Existing methods for estimating Mutual Information and Entropy}
\label{mutualmeasures}

Since the underlying probability distributions are unknown, the Mutual information based measure has to be estimated. This can be done by using a non parametric approach, or by a parametric method. Parametric methods need specific form of stochastic processes as stated by \cite{physicaA:dionisio} ,\cite{granger:2000}. \cite{poczos} have given a non-parametric method for estimating mutual information based on lengths of edges of nearest-neighbour graph. However, this method sometimes gives very small negative values. This is not desirable since estimated value of dependence measure should be non-negative.

Histogram based density estimates can be used for estimating mutual information and entropy. \cite{Jenssen} have shown that Parzen window-based estimators for the quadratic information measures have a dual interpretation as Mercer kernel-based measures, where they are expressed as functions of mean values in the Mercer kernel feature space. They have shown Mercer kernel and the Parzen window to be equivalent.

\cite{Kong14} have shown that objective function of Linear discriminant analysis is sum of KL-divergence between two corresponding classes for each pair of points. They also proposed Pairwise-Covariance Linear Discriminant Analysis where the objective function is modified to emphasize the contribution of pair of classes which are separated by a smaller distance. \cite{Sahami96learninglimited} has used mutual information and class conditional mutual information for designing a $k-dependence$ Bayesian classifier.

\cite{reshef} proposed a method to calculate dependence using mutual information, wherein histogram was used to estimate density. Equi-probable partitions are created along one axis and the other axis is partitioned so that the overall value of the measure is maximized. Number of partitions is limited by $n^{0.6}$, where $n$ is the number of points in data set. This process is repeated by swapping the axes and maximum of either is taken. Thus, Reshef et al. found bin length for calculating mutual information using an elaborate process. Note that this process of finding the bin length is computationally expensive. As the partition is chosen with a preference for higher values of estimate the value attained for independent variables is relatively high, and hence it generally overestimates the dependency value as stated by \cite {Tibshirani}.

In this article, we propose a fast and implementable method for finding non-linear relationship between variables quickly using an estimate of mutual information based measure. This estimate is also calculated using histogram based density estimation. However, we have used a partitioning scheme which leads to better performance in case of noisy data. The partitioning is done using statistical properties of data which can be calculated easily. No optimization is required over different partitions of same data. This makes the proposed method to be efficient in terms of execution time and memory. Consequently, it can be used for large volumes of data where applying existing methods becomes infeasible due to time and memory constraints.

\subsection{Entropy and mutual information: Definitions}
For two discrete random variables $X$ and $Y$, let the probability mass functions be denoted by $P_1(x)$ and $P_2(y)$, and their joint probability mass function be denoted by $P(x,y)$. For two continuous random variables $X$ and $Y$, let the probability density functions be denoted by $p_1(x)$ and $p_2(y)$, and their joint probability density function be denoted by $p(x,y)$.

For a discrete random variable $X$, the term $H(X)$ is called entropy and it is defined as $\sum_x P_1(x)log (\frac {1}{P_1(x)})$. For two discrete random variables $X$ and $Y$, the term $H(X|Y)$ is called conditional entropy, and it is defined as $\sum_y \sum_x P(x,y)log(\frac {P_2(y)}{P(x,y)})$. The term $H(X,Y)$ is called joint entropy and is defined as $\sum_y \sum_x P(x,y)log(\frac {1}{P(x,y)})$. Similar definitions for $H(X)$, $H(X|Y)$, and $H(X,Y)$ can be given when $X$ and $Y$ are continuous random variables.

Mutual information I(X,Y) between variables X and Y measures the decrease of uncertainty about X caused by the knowledge of Y, which is the same as the decrease of uncertainty about Y caused by the knowledge of X. It measures the amount of information about X contained in Y, or vice versa . The definition is taken from a book by \cite{Ash}.
\begin{align}
I(X,Y)  &= H(X) - H(X|Y ) \\
        &= H(Y ) - H(Y |X) \\
        &= H(X) + H(Y ) - H(X, Y )
\end{align}
 It may be noted that mutual information of a discrete random variable with itself is its entropy. This is also referred to as self information. Entropy is sometimes also described as amount of uncertainty in the system. \cite{JaynesMaxEnt} has stated that in statistical modeling one should choose a system with maximum entropy possible under given constraints. \cite{XiZhaoMaxEnt} have proposed a scheme for refinement of fuzzy if-then rules based on the maximization of fuzzy entropy on the training set. On similar lines, \cite{XiZhaoMaxAmbi} have proposed maximum ambiguity based sample selection in fuzzy decision tree induction. \cite{XiZhaoGenFuzzy} have studied the relationships between generalization capabilities and fuzziness of fuzzy classifiers. They have experimentally confirmed the relationship which meets the conditional maximum entropy principle.

\subsection{Existing Mutual Information based measures}
\label{exisMImeas}

\cite{Tambakis} presents a mutual information estimator, which is based on equidistant cells. The author suggests  the determination of the Self-Information Measure (SIM), through the univariate non-parametric predictability computation, as a function of the mutual information and the number of partitions (K), that is:
\begin{align}
SIM_{K}=\frac{I(X,Y)}{log(K)}
\end{align}
\cite{reshef} propose another way of estimating dependence measure as:
\begin{align}
MINE_{k,m}=\frac{I(X,Y)}{\log(\min(k,m))},
\end{align}
where $k$ and $m$ are numbers of partitions along each direction.

\cite{Kvalseth} described three measures in his paper, and they are given below.
\begin{align}
\label{eqKM1first}
KM1=\frac{I(X,Y)}{\min(H(X),H(Y))}\\
KM2=\frac{I(X,Y)}{\max(H(X),H(Y))}\\
KM3=\frac{2*I(X,Y)}{(H(X)+H(Y))}
\end{align}
$KM2$ and $KM3$ were proposed by \cite{Horibe} and \cite{Norusis} respectively.  $KM2$ has been used by \cite{Kraskov:2008} for clustering. \cite{Kvalseth} mentions that the value of all three measures is $1$ in case of strict one to one association between the variables. However, for $KM1$ this is a sufficient condition but not a necessary condition. In this article, we shall be using $KM1$.

\subsection{Other well known dependence measures}

Among the well known dependence measures, Pearson correlation coefficient is a very widely used measure for computing linear relationship.
It is invariant to scaling and translation. But it does not detect non linear relationships.

Spearman correlation between two variables $X$ and $Y$ is Pearson correlation between rank of points of $X$ and $Y$. Therefore it detects relationship where the variables are monotonically related to each other, even if the relation is non linear. It will not be able to detect non monotonic relations like sinusoidal wave, saw-tooth wave etc.

Dependence measure called distance correlation (dcor) defined by \cite{SRpaper} is correlation between linear functions of interpoint distances. The authors have defined distance covariance as $\nu _n^2(X,Y)=1/n^2 \Sigma_{k,l=1}^n A_{k,l} B_{k,l}$, where $A_{k,l}$ and $B_{k,l}$ are linear functions of pairwise distances between points of $X$ and $Y$. Similarly, distance variance is defined as $\nu_n^2(X)=1/n^2 \Sigma_{k,l=1}^n A_{k,l}^2$, where $A_{k,l}$ is same as mentioned earlier. Distance correlation($\mathcal{R}$) is normalized distance covariance defined by $\mathcal{R}^2(X,Y)=\frac{\nu_n^2(X,Y)}{\sqrt{\nu^2(X) \nu^2(y)}}$ if $\nu^2(X) \nu^2(y)>0$, and 0 otherwise. It has been shown that distance covariance can also be defined as $\nu_n^2=||f_{X,Y}^n(t,s)-f_X^n(t) f_Y^n(s)||^2$, where $f_{X,Y}^n(t,s)$ is the joint empirical characteristic function of the sample ${(X,Y), \cdots(X_n,Y_n)}$ and $f_X^n(t)$, $f_Y^n(s)$ are marginal characteristic functions of $X$ and $Y$. Distance correlation detects relationships where both the variables are determined by each other (i.e., one one and onto function exists between the values of two variables.), and therefore it does not detect relationships like sinusoidal wave, circle etc.

At some places calculating the dependence between variables is not explicitly required but understanding the dependence remains an important part of the problem. For example, identifying a low rank space which is common to multiple classifiers allows us to couple multiple learning tasks using the basis of low rank subspace. Such an algorithm based on linear multi task learning has been designed by \cite{Chen2012}. They used regularization for enforcing sparsity and applied low rank constraint to the objective function to encourage a low rank structure.

\section{Proposed estimates of dependence measure : properties and implementation}
\label{proposedSec}

\subsection{Dependence measure estimated in this paper: Properties and their implications for proposed estimates}
\label{chosenMeas}

It can be shown that $KM1$ attains a value of $1$ if a strict one way association exists, i.e. either of the variables is a function of other. This allows us to detect functional forms having one way dependence e.g. $y=x^2$, $y=\sin(x)$. Very little work has been done using $KM1$, though it was proposed long ago. We feel that $KM1$ as a measure of association between variables has many desirable properties. Though \cite{Kvalseth} describes $KM1$ briefly but no method has been proposed to estimate the defined measure. In the proposed method the values of $I(X,Y), H(X)$ and $H(Y)$ are estimated using a new density estimation method described in sections \ref{binlength} and \ref{MIChist}. The resulting estimates are then used to calculate an estimate of $KM1$ which is named as Mutual Information based Dependence Index ($MIDI$).

\cite{pompe:1998}  presents some important properties of mutual information in the discrete case, namely:
\begin{align}
&I (X,Y) = 0      &&\textit{,iff X, Y are statistically independent}\\
&I (X,Y) = H (X)  &&\textit{,iff X is a function of Y}\\
&I (X,Y) = H (Y ) &&\textit{,iff Y is a function of X}
\end{align}
The measure estimated in this paper is
\begin{align}
\label{eqKM1second}
KM1=\frac{I(X,Y)}{min(H(X),H(Y))}
\end{align}

It can be seen that KM1 attains a value of $1$ only when X is a function of Y, or Y is a function of X or both. As, $I(X,Y)=0$ if and only if X and Y are statistically independent, KM1 attains a value $0$ in case of statistical independence. The value of $KM1$ always lies between $0$ and $1$. As the proposed density estimates are shown to be $L_1$ consistent in following sections it can be said that the value of proposed estimate MIDI goes to $1$ when a perfect one way relationship exists between two variables and $0$ when two variables are completely independent, as number of points increases.

In this paper the entropy and mutual information of variables are estimated by dividing each axis into bins. The proportion of number of points lying in the bin to total number of points is used as the value of Probability mass function for points lying in the bin. Since, mutual information and entropy for discrete random variables are invariant under change of variable the value of calculated estimate remains unchanged as long as structure of bins remains the same. As the ratio of maximum separation between the points along a given axis and range along this axis are independent of scale, the division of each axis is invariant to scaling. The proposed estimate is therefore invariant to scaling. Similarly, the bin structure and the proposed estimate are also invariant to translation. Under above mentioned conditions the value of estimate remains unchanged. The estimate however is not invariant to rotation or any transform that is not order preserving.

\subsection{Determining bin length and calculating dependence}
\label{BinReq}
If histogram approach for density estimation is used, the length of bin becomes a very important parameter affecting the value of estimate as stated by \cite{Dionisio}. \cite{DarbellayVajda} mention the importance of selecting correct partitions for estimation of mutual information using histogram method and recommend the use of data dependent procedures for determining appropriate partitions.

If the bin length $\epsilon$ is too large the index fails to capture non linear relationships as it will not find unique mapping between pair of axes. So it is required that $ \lim_{n \to \infty} \epsilon \to 0$. On the other hand if bin length is too small the distribution will not seem to be continuous. As the number of points increases, the bin length should decrease. However, if the length of bin decreases too quickly with increasing $n$, many partitions will remain empty. \cite{parzen} ,\cite{Lugosi} have given criteria for selection of appropriate bin size.

In this article a new method is proposed for finding bin length. It is computationally simple, and the value of estimated measure is closer to zero than the values obtained by proposed by MINE, for the case of independent variables. The resulting estimate is also less susceptible to noise.

In the following paragraphs the proposed method for calculating dependence is given. Entropy and mutual information of the variables are calculated by dividing each axis into bins. The length of the bin on one axis is determined using a function of connectivity distance along that axis. On the other axis the number of bins is proportional to logarithm of number of data points. The proportion of number of points lying in the bin to the total number of points, is used as the value of Probability mass function for points lying in the bin. These estimates are then used to calculate the mutual information, entropy along first dimension and entropy along second dimension. The calculated values are used to obtain $MIDI$. The procedure is repeated taking the bin length to be a function of connectivity distance on second axis and a function of logarithm of number of data points on first. Higher of two values is taken to be the calculated dependence index. Theoretical properties of the method applied for density estimation have also been discussed.

\subsubsection{Proposed method of bin length estimation}

\label{binlength}
\cite{alpha:hull} have used the average edge length of Minimal spanning tree as the value of $\alpha$, the radius of the area for determining points lying nearby for set estimation. Here we propose to use the connectivity distance as defined by \cite{appelRusso}, and \cite{penrose:strongLaw} to determine the bin length along one dimension. We use a function of connectivity distance along a single dimension. This connectivity distance is also called maximal spacing by \cite{Slud}.

Let the given data be $\{(x_1,y_1), \cdots, (x_n,y_n)\} \subset {\mathbf{R}}^2$. We assume that the points are drawn
i.i.d. from set $A$ following a continuous probability density function $f$, which is unknown. Let
$A$ be path connected, compact, $Closure(Interior(A))=A$ and boundary of A, denoted by $\delta(A)$ is such that $\lambda(\delta A)=0$, $\lambda$ is Lebesgue measure. Let the support of $f$ be $A$, i.e. $f(x)=0$ $ \forall x \in A'$, where $A'$ is complement of $A$ and $\int_v f(x) dx > 0$  $\forall v$, $v$ open and $v\cap A \not=\phi $.

The data is sorted along x axis. Let the sorted values be $x_{(1)}\cdots x_{(n)}$. The maximum difference between consecutive ordered points be $L_{max}=\sup\{x_{(i+1)}-x_{i}:1\leq i\leq n-1\}$, where $n$ is number of points in sample. The proposed bin length along x axis is $n^cL_{max}$. Here $c$ is a constant such that $0<c<1$. For a fixed n, as the value of c goes towards $0$ we get smaller bin lengths. Since the number of bins increases we get a more detailed picture so we are more likely to get unique mappings leading to higher value for dependence measure estimate. However, the decrease in bin length can also increase value of index for uniform distribution although at a smaller rate. Along the y axis  number of divisions is given by $d_n=log_{10}(n)$. Once the bin length is found, estimate of dependence measure can be calculated as explained in section \ref{MIChist}. The same procedure needs to be repeated after swapping the x and y axes and the maximum of the two estimates is taken as the final value.

\subsubsection{Calculating estimated mutual information based dependence measure using proposed histogram scheme: Calculating MIDI}
\label{MIChist}

Let the range of data along x and y axes be scaled to $[0,1]$. The origin of the interval is taken as starting point for histogram calculation. The mutual information according to densities estimated using the above mentioned scheme is given by
\begin{align}
\hat{I}=\sum_{i \in \pi_x}\sum_{j \in \pi_y} \frac{\#_{i,j}}{n}\log \left [\frac{n\#_{i,j}}{\#_i \#_j}\right ]
\end{align}
where $\pi_x$ and $\pi_y$ are partitions along x and y directions respectively. The number of points in the cell given by intersection of interval $i$ along x axis and $j$ along $y$ axis is given by $\#_{i,j}$. Number of points in interval $i$ along x axis is given by $\#_{i}$. Similarly, number of points in interval $j$ along y axis is given by $\#_{j}$. Entropy values along x and y axes are calculated in similar way as $\hat{H}_x=\sum_{i\in \pi_x}\frac{\#_i}{n}\log(\frac{n}{\#_i})$ and $\hat{H}_y=\sum_{j \in \pi_y}\frac{\#_j}{n}\log(\frac{n}{\#_j})$. $MIDI$ is now calculated as $\hat{I}/min(\hat{H}_x,\hat{H}_y)$. This $MIDI$ is the estimated form of $KM1$ given in equations (\ref{eqKM1first}) and (\ref{eqKM1second}).

\subsection{$L_1$ consistency of density estimates obtained using proposed bin length}
\label{proposedL1}

For consistency of density estimates using histogram method, the length of bins is very important parameter. \cite{Lugosi} gave general sufficient conditions for the strong $L_1-consistency$ of histogram density estimates based on data-dependent partitions. The relevant theorem from the article has been reproduced in appendix as Theorem \ref{L1cons}. In each case, the desired consistency requires shrinking cells, sub-exponential growth of a combinatorial complexity measure, and sub-linear growth of the number of cells. Under the assumptions of the theorem the $L_1-error$ of the normalized partitioning density estimates converges to zero with probability one.

It has been seen  from Lemmas \ref{SubLinGrow} and \ref{lemSubExpo} given below, that the number of cells is sub linear and the combinatorial complexity is sub exponential as $n \to \infty$ for the proposed method.

\begin{lemma}
\label{SubLinGrow}
In the proposed histogram scheme, for a fixed c, $ n^{-1}m(A_n) \to 0 \text{ ,as } n \to \infty$
\end{lemma}
\begin{proof}
In the proposed scheme the smallest possible value of $L_{max}$ will be attained when all points are equidistant to each other along x axis. In this case the value for $L_{max}$ will be $g/(n-1)$, where $g$ is the range along x axis. In this case the number of divisions along x axis will be $\frac{g}{n^c.L_{max}}=\frac{(n-1)}{n^c}$. Number of divisions along y axis is given as $d_n=log_{10}(n)$. So the maximum number of cells in given scheme is $\frac{(n-1).d_n}{n^c}=\frac{(n-1)log_{10}(n)}{n^c}$ where $0<c<1$.
\begin{align}
m(A_n)      &\leq\frac{(n-1)log_{10}(n)}{n^c}\\
n^{-1}m(A_n)&\leq\frac{(n-1)log_{10}(n)}{n^{(c+1)}}\\
            &<\frac{n.log_{10}(n)}{n^{(c+1)}}\\
            &=\frac{log_{10}(e).log(n)}{n^c}
\end{align}
Note that $\frac{log(n)}{n^c} \to 0$ for a fixed $c$ for $n \to \infty$
\begin{align}
\implies  n^{-1}m(A_n)&=0 \textit{ ,as } n\to\infty
\end{align}
\qed
\end{proof}

\begin{lemma}
\label{lemSubExpo}
For the proposed histogram scheme $n^{-1}\log(\Delta_n^\ast (\mathcal{A}_n))\to 0$ as $n\to\infty$
\end{lemma}
\begin{proof}
The maximum number of divisions along x axis is given by $\frac{(n-1)}{n^c}$. Note that $\frac{(n-1)}{n^c}<n^{1-c}$. Let $\lceil x\rceil$ denote the ceiling of positive real number x. So maximum number of distinct partitions along x axis is bounded by $\binom{\lceil n+n^{1-c} \rceil}{\lceil n^{1-c} \rceil}$. Similarly, the maximum number of divisions along y axis being $\ log_{10}(n)$, maximum number of distinct partitions along y-axis is bounded by $\binom{\lceil n+log_{10}(n) \rceil}{\lceil log_{10}(n) \rceil}$. So maximum number of distinct partitions considering both the dimensions, $\Delta_n^\ast (\mathcal{A}_n)$ is bounded by the product $\binom{\lceil n+log_{10}(n) \rceil}{\lceil log_{10}(n) \rceil} \binom{\lceil n+n^{1-c}\rceil}{\lceil n^{1-c} \rceil}$, where $0<c<1$.
\begin{align}
n^{-1}\log(\Delta_n^\ast (\mathcal{A}_n))\\
\leq n^{-1}\log\left[\binom{\lceil n+\log_{10}(n) \rceil}{\lceil \log_{10}(n) \rceil}\binom{\lceil n+n^{1-c}\rceil}{\lceil n^{1-c}\rceil}\right]\\
=\frac{1}{n}\log\left[{\binom{\lceil n+log_{10}(n)\rceil}{\lceil log_{10}(n)\rceil}}\right]+\frac{1}{n}\log\left[{\binom{\lceil n+n^{1-c}\rceil}{\lceil n^{1-c}\rceil}}\right]
\end{align}
 The upper bound of $\log\left[\binom{n}{k}\right]$ is given by \cite{CsiszarKorner} as $nH(k/n)$, where $H(\epsilon)$ represents the binary entropy of $\epsilon$, defined by $H(\epsilon)=-\epsilon\log(1-\epsilon)-(1-\epsilon)\log(\epsilon)$ . It follows that
\begin{align}
n^{-1}\log(\Delta_n^\ast (\mathcal{A}_n))
\label{eqLemsubexpo}
\leq \frac{n+\log_{10}(n)}{n}H(\frac{\log_{10}(n)}{n+\log_{10}(n)})
    +\frac{n+n^{1-c}}{n}H(\frac{n^{1-c}}{n+n^{1-c}})\\
\label{eqentpterm1}
\text{However, as }n \to \infty,\text{ } \frac{\log(n)}{n+\log_{10}(n)}\to 0 \\
\label{eqentpterm2}
\text{ As }n \to \infty,\frac{n^{1-c}}{n+n^{1-c}} \to 0 \text{, where } 0<c<1\\
\text{As }n \to \infty,\text{ } \frac{n+\log_{10}(n)}{n}\to 1 \\
\text{ As }n \to \infty, \frac{n+n^{1-c}}{n}\to 1 \text{, where } 0<c<1
\end{align}
 Note that $H$ is increasing on $(0,1/2]$, H is symmetric about $1/2$, and $H(\epsilon)\to0$ as $\epsilon \to 0$. Therefore, entropies of terms given in (\ref{eqentpterm1}) and (\ref{eqentpterm2}) vanish as $n \to \infty$. So both terms on the right hand side of inequality (\ref{eqLemsubexpo}) vanish as $n\to\infty$. We find that $n^{-1}\log(\Delta_n^\ast (\mathcal{A}_n))=0$ as $n\to\infty$.
\qed
\end{proof}

Lemmas \ref{shrinkCellUni} and \ref{shrinkCell} show that diameter of each cell shrinks to $0$ as $n \to \infty$ for the cases of uniform and non-uniform distributions respectively. For this purpose limit results for maximal spacings have been used. The largest spacings for independent identically distributed (i.i.d.) uniform sequences in [0, 1] have been studied by several authors namely \cite{Levy}, \cite{Darling}, \cite{Slud} and \cite{LimUni}. In this article a theorem by \cite{Darling} (can be found in appendix as theorem \ref{Darlingthm}) has been used to prove lemma \ref{shrinkCellUni}. \\

\begin{lemma}
\label{shrinkCellUni}
For proposed histogram scheme, $\mu\{x : diam(\pi_n[x]) > \gamma \} \to 0$, with probability one for every $\gamma >0$, in case of uniform distribution.
\end{lemma}
\begin{proof}
By theorem \ref{Darlingthm}
\begin{align}
| M_n-\log (n)/n | &= C(\log (\log (n))/n)
\end{align}
Multiplying both sides by $n^c$, where c is a constant such that $0<c<1$, for $n \to \infty$
\begin{align}
| n^c.M_n-\log (n)/n^{1-c} | = C(\log (\log (n))/n^{1-c})
\end{align}
Note that, $\log n/n^{1-c}\to 0$, as n goes to infinity. Also, $\log (\log (n))/n^{1-c}\to 0$, as n goes to infinity. So $n^cM_n \to 0$ as $n \to \infty$. It may be noted here that $L_{max}$ is same as $g.M_n$ where $g$ is the range. So $n^cL_{max} \to 0$ as $n \to \infty$, in case of uniform density.
\qed
\end{proof}

Similarly, limit results for ordered spacings (where uniform distribution has not been assumed) can be found in article by \cite{Deheuvels}. Theorems \ref{dev1} and \ref{dev2} proposed by \cite{Deheuvels} have been used to show that the diameter of each cell shrinks to $0$ in the given scheme as $n\to \infty$ for non-uniform distributions in Lemma \ref{shrinkCell}. These two theorems proposed by \cite{Deheuvels} have been stated in the appendix. Using these theorems we state the following lemma.

\begin{lemma}
\label{shrinkCell}
For proposed histogram scheme, $\mu\{x : diam(\pi_n[x]) > \gamma \} \to 0$, with probability one for every $\gamma >0$
\end{lemma}
\begin{proof}
We note that, if the conditions stated in theorems \ref{dev1} and \ref{dev2} given by \cite{Deheuvels} and stated in appendix are fulfilled,
 \begin{align}
\lim{\sup_{n\to\infty}{\frac{nM_k^{(n)}f(x_0)-\log(n)}{\log_2(n)}}}=\frac{2}{k}-\frac{1}{r}\\
\label{eqshcell1}
\implies \lim{\sup_{n\to\infty}{n^c M_k^{(n)}}}=\frac{(\frac{2}{k}-\frac{1}{r})\log_2(n)+\log(n)}{n^{1-c} f(x_0)}
\end{align}
Note that, for $0<c<1$
\begin{align}
\label{eqshcell2}
\lim_{n \to \infty}log(n)/n^{1-c}=0\\
\text{Also,}
\label{eqshcell3}
\lim_{n \to \infty}(2/k-1/r)log_2(n)/n^{1-c}=0
\end{align}

Maximal spacing between the points is given by $L_{max}=g.M_1^{(n)}$, where g is the range along given dimension. It follows from (\ref{eqshcell1}), (\ref{eqshcell2}) and (\ref{eqshcell3}) that, $n^cL_{max}\to0$ as n goes to $\infty$, under the given conditions. $n^cL_{max}$ is taken as length of each cell along one axis then along the other axis say y-axis the length of each division is taken as $g_y/\log_{10}(n)$ where $g_y$ is the range along y axis. This length too goes to $0$ as $n$ goes to infinity. As length along each direction goes to $0$  as $n$ goes to infinity we can say that diameter will shrink to $0$ as $n \to \infty$.
\qed
\end{proof}

As a consequence of theorem \ref{L1cons} given in appendix , and lemmas \ref{SubLinGrow}, \ref{lemSubExpo}, and \ref{shrinkCell} given above, the following theorem can be stated.
\begin{theorem}
\label{L1consproposedthm}
The proposed density estimate is $L_1$ consistent.
\qed
\end{theorem}

Though, we cannot say that the proposed method is the best way of getting bin length, the proposed method satisfies the requirements of $L_1$ $consistency$ as defined by \cite{Lugosi}. It is a simple and efficient method as compared to the existing methods. The algorithm for calculating the proposed estimate of dependence measure for a dataset is given below.

\subsection{Algorithm}
\label{Algo}
\begin{enumerate}
  \item{}Sort the data along x-axis. Find the maximum separation between consecutive points along x-axis. Let this be represented by the value $S$.
  \item{}Let the bin size along x axis be $B_x= n^c.S$, where $n$ is the total number of data points and $c$ is a constant $0<c<1$.
  \item{}Divide range y-axis in $log_{10}(n)$ bins of equal length.
  \item{}For each region created by intersection of division $i$ along x axis and $j$ along y  axis, calculate the number of points as $\#_{i,j}$.
  \item{}For each region $i$ along x-axis calculate the number of points as $\#_{i}$.
  \item{}For each region $j$ along y-axis calculate the number of points as $\#_{j}$.
  \item{}Calculate the value $\hat{I}=\sum_{i=1}^{n_x} \sum_{j=1}^{n_y} \frac{\#_{i,j}}{n}\log(\frac{n \#_{i,j}}{\#_{i} \#_{j}})$, where $n_x$ and $n_y$ are number of divisions along each direction.
  \item{}Calculate the value of $\hat{H}_x=\sum_{i=1}^{n_x}\frac{\#_i}{n}\log(\frac{n}{\#_i})$.
  \item{}Calculate the value of $\hat{H}_y=\sum_{j=1}^{n_y}\frac{\#_j}{n}\log(\frac{n}{\#_j})$.
  \item{}Calculate $MIDI_x = \frac{MI}{min(\hat{H}_x,\hat{H}_y)}$
  \item{}Repeat the procedure after swapping x and y axes and take maximum of the two results $MIDI=max(MIDI_x,MIDI_y)$.
\end{enumerate}

\section{Results}
\label{results}
In order to provide experimental results using the proposed method, and comparing them with other indices of association, several datasets have been considered. While conducting the experiments, value of c is taken to be $0.1$. The chosen value of $c$ gives low values of MIDI for uniform distribution and is capable of detecting non linear relationships as demonstrated in this article. Essentially, artificial datasets have been considered here since existing relationship between variables if any is known. This makes judgement of results easier. Experimental results on one real life dataset have been reported here, where in the existence of relationship beteen variables is known.

\subsection{Artificial datasets}
There exist three possible types of relationship between two variables. These are
\begin{enumerate}
\item{}Relationship between two variables $X$ and $Y$ where either X is a function of Y or vice versa.
\item{}Relationship where one to one mapping does not exist but there is a functional form (e.g., cartesian coordinates for circle) between them.
\item{}Relationship is not functional (e.g., Two normal distributions with correlation coefficient $\rho$ where $|\rho|<1$).
\end{enumerate}
MINE given by \cite{reshef}, distance correlation (dcor) given by \cite{SRpaper} have been used for comparison because these methods are well studied and give good results for many data sets as compared to other methods like Spearman correlation coefficient and Pearson correlation coefficient. A comparative study of MINE and dcor is also done by \cite{Tibshirani}. MINE, dcor and MIDI have been calculated for the data sets described in Table \ref{artData} and for data drawn from normal distribution with different values of correlation coefficient.

It may be noted that the value of MINE is higher than the value obtained by MIDI even in case of uniform and gaussian distributions where the variables are generated independently and are expected to be completely uncorrelated. This happens because MINE chooses grid size which maximizes the value of the measure. Also, MIDI is able to detect relationships which remain undetected by dcor.

\subsubsection{Noiseless  data}
\label{noiselessdatapar}
Three types of noiseless data are considered here. Table \ref{artData} provides the functional forms corresponding to these three different associations stated earlier. Note that the same functional forms were used by Reshef et al to illustrate MINE for noiseless data. The generation of data corresponding to the functional forms of  table \ref{artData} is done in the following way.  The points for variable $X$ are drawn randomly from uniform distribution and the value of $Y$ is calculated as a function of $X$. Different sample sizes are considered, and the estimated indices corresponding to these sizes are given in table \ref{noislessDat}. It may be noted in table \ref{noislessDat} that the values of indices are not given for each method and for each size. MIDI is calculated for data sets of sizes 1000, 2000, 5000 and 10000. MINE is calculated for datasets of sizes 1000, 2000 and 5000. Distance correlation is calculated for datasets of sizes 1000, 2000. MINE and distance correlation are calculated for smaller data sets as bigger data sets have very high memory and time requirements. The experiments have been executed on an Intel Pentium D 925 3.00GHz CPU and 1.5 GB memory. As the sample size increases, value of MIDI goes closer to 1 for variables which are related to each other and goes to 0 when the variables are independent.

For comparing the performances in case of non functional relationship, Table \ref{normDat} provides the estimated values of the three indices for different sample sizes in case of bivariate normal distribution with mean 0 and  variance 1 for both $X$ and $Y$, and different values of correlation coefficient($\rho=1,0.95,0.9,0.8,0.7,0.6,0.5,0.3,0.01$). MINE and dcor are not calculated for large data sets as they have very high time and memory requirements. We find that the value of MIDI falls with decrease in correlation coefficient. With increase in number of points, value of MIDI remains similar when correlation coefficient is high but drops significantly when correlation coefficient is small.

\subsubsection{Noisy data}
Another aspect under consideration for judging the quality of results is noise. As more noise is introduced in a dataset, existing relationship between variables is diluted, and thus the value of the index drops. Hence, different datasets with different noise levels are also considered. Due to high time and memory requirements of MINE and dcor these indices are not calculated for large data sets.

To calculate the values of indices, points are drawn from uniform distribution with mean $\mu=0$ and variance $\sigma^2$. The different values of $\sigma^2$ used for generating noise are $10^{-6},10^{-4},10^{-2},10^{-1},1,10^1,10^2$. The noise is added to the variable Y and values of different indices are calculated. The tables \ref{MIDInoisy},\ref{MIDInoisytwo}, \ref{MIDInoisyfive} and \ref{MIDInoisyten} report the values of indices for different sizes of data sets for functional relationships given in table \ref{artData}.  It might be noted that with increasing number of points the values of MIDI for low noise data becomes closer to 1 when relationship exists, and moves to 0 when there is no relationship between the variable. The value of MIDI falls smoothly with increasing noise level.

For studying the Non functional relationship, X and Y are considered to follow bivariate normal distribution as described in section \ref{noiselessdatapar}. Varying levels of noise are added to $Y$ as described in previous paragraph. The results are reported in tables \ref{MIDInoisynorm}, \ref{MIDInoisytwonorm}, \ref{MIDInoisyfivenorm} and \ref{MIDInoisytennorm}. We find that value of MIDI decreases with decreasing correlation coefficient. The value of MIDI also decreases with increase in noise level.

\subsection{Power of different methods}
Yet, another aspect under consideration is the Power of the test where the alternative hypothesis is the non existence of association. If the Power is high, then the statistic under consideration is satisfactory. \cite{Tibshirani} used some functions for measuring the Power. Their functions are used (Table \ref{PowerFuncs}) so that a valid comparison can be provided. The methodology used for comparing power is also same as \cite{Tibshirani}

The results of power of test against alternative hypothesis (non existence of relationship) for MINE, dcor and correlation coefficient have been published by \cite{Tibshirani}. Here the power of all the three indices shown in Fig. \ref{MIDIpower} are calculated by the same method as outlined here. The functions for which power is calculated are given in Table \ref{PowerFuncs}. For each function the value of different indices is calculated for X and Y, where X consists of 1000 points drawn from uniform distribution. Y is generated using the given function. Noise is added to Y. The noise includes points drawn from normal distribution with mean 0 and varying values of variance for different noise levels. 30 different standard deviation levels are considered with $\sigma$ being product of noise scale mentioned in table \ref{PowerFuncs} and $1/10, 2/10 \cdots,30/10$. The noise levels are same as used by \cite{Tibshirani}. For example, for the the case of cubic function, the $\sigma$ values are $10/10, 20/10, 30/10, \cdots ,300/10$. Thus, for each functional form and for each $\sigma$, 500 data sets are generated, each having 1000 points, and values of different indices are calculated. Now the values of the indices are calculated for 500 data sets where no relationship exists by drawing both X and Y from uniform distribution. Thus, we get the values of indices for null scenario. The cut off value of the index is calculated as 0.95 quantile of indices generated for null scenario. Amongst the data sets where a particular functional relationship exists, the ratio of number of data sets for which the value of index is above the cut off for a given $\sigma$ to total number of data sets generated for the same $\sigma$ is taken as power at the particular noise level for relationship under consideration.

In case of linear relationship, distance correlation works better than proposed method and MINE. In all the seven non-linear cases considered, the proposed method is seen to have more power than dcor most of the times. In a way, this was expected. The reason is that dcor is biased towards linear relationship. Distance correlation reduces to correlation coefficient between two variables if the distributions are normal. So, it is expected that linear relationship is better captured by dcor as compared to proposed method.\\
It may be noted that the computational cost of proposed method is very small as compared to other existing methods. The proposed method consumes less time and memory as compared to dcor and MINE, making it easier to use with large datasets.

\subsection{Experiments on real life data set}
The proposed method was used to analyze yeast gene expression data set given by \cite{Spellman} to identify genes whose transcript level varies periodically within the cell cycle. MIDI could identify time regulated genes which remained unidentified by dcor and MINE. Table \ref{realdata} lists some genes which are found to be time regulated by several studies like \cite{SPECTRANS}. MIDI obtained a high value for dependence between gene expression and time for these genes. Rows 1 through 33 show time regulated genes for which only MIDI obtained high values (greater than 0.8), while dcor and MINE obtained low values (less than 0.3).We found that the proposed method is giving significantly high values for several time regulated genes as compared to dcor and MINE. Rows 34, 35 and 36 give examples of genes which are identified to be time regulated by all three methods, MINE, MIDI and dcor. Note that this list is not exhaustive. The table shows examples where MIDI works better than the other two. However, out of the total of 4381 genes (all of them are not known to be time regulated), MIDI obtained values less than either MINE or dcor in 68 cases. For all other $98.45\%$ genes, MIDI values are higher than both MINE and dcor.

\section{Conclusion and Future work}
\label{conc}
This paper introduces a computational method to estimate non linear dependence between two variables. A function of connectivity distance has been used to determine the bin length for calculating the density estimates. The density estimates obtained by this method are shown to be strongly consistent. These density estimates are used to calculate Mutual Information based Dependence Index (MIDI), which is an estimate of an existing dependence measure. Through several experiments we find that, the value of proposed index goes to its true value as number of points increases. We also find that the value of MIDI goes down smoothly if noise is added gradually to variables having perfect dependence. The computational cost for MIDI is low in comparison with both MINE and dcor. MIDI overcomes the problem of over-estimation which occurs with MINE. MIDI is also capable of detecting non linear relationships where either of the variables is a function of other unlike dcor which requires the relation to be one to one and onto. The proposed method is found to work better on a real life data set than competing methods.

It may be noted that histogram based estimation is a type of kernel density estimation. Unlike the more commonly used Gaussian kernel the histograms have the advantage of having a compact base, which is relevant for proposed method of estimating dependence. In the future, experiments can be performed with other types of compact kernels. Here we have chosen to use the basic histogram method because it has the advantage of simplicity which results in quick execution.

\FloatBarrier

\section{Tables and figure}
\begin{table*}[htb!]
\renewcommand{\arraystretch}{1.3}
\caption{Different datasets used for experiments}
\label{artData}
\begin{tabular}{|c|c|c|}

  \hline
  Function & Description&Domain \\
  \hline
  line & $y=x$& $x \in [0,1]$ \\
  half-parabola & $y=x^2$ & $x \in [0,1]$ \\
  parabola & $y=(x-0.5)^2$ & $ x \in [0,1]$ \\
  exponential &$y=10^x$&$ x \in[0,1]$\\
  sinusoidal &$y=\sin(10\pi x +x)$ & $ x \in[0,1]$\\
  sinusoidal (fourier frequency)[sff]&$y=\sin(16\pi x)$&$ x \in [0,1]$\\
  sinusoidal (non fourier frequency)[snff]&$y=\sin(13\pi x )$&$x \in [0,1]$\\
  sinusoidal (varying frequency)[svf]&$y=\sin(7\pi x (1+x))$&$ x \in [0,1]$\\
  circle&\parbox{5cm}{$y=(2*z-1) * (\sqrt{1 - (2*x - 1)^2})$, where $z$ is randomly chosen from $\{0,1\}$} &$ x \in [0,1]$\\
  normal uncorrelated &$\mu_1=0$, $\sigma_1=1$, $\mu_2=0$, $\sigma_2=1$, $\rho=0$& $(x,y) \in (-\infty, \infty) \times (-\infty, \infty)$ \\
  uniform &random number generator & $(x,y) \in [0,1]\times[0,1]$ \\

  \hline
\end{tabular}
\end{table*}

\begin{table*}
\renewcommand{\arraystretch}{1.3}
\caption{Values of dependence measure estimates on artificial noiseless data sets with varying sizes}
\label{noislessDat}
\begin{tabular}{|c||c|c|c||c|c|c||c|c||c|}

  \hline
  Number of points&\multicolumn{3}{|c||}{1000}&\multicolumn{3}{|c||}{2000}&\multicolumn{2}{|c||}{5000}&10000\\
  \hline
  Function & MINE &Dcor & proposed &MINE&Dcor&proposed&MINE&proposed&proposed\\
  \hline
  line & 1.0000 &1.0000& 0.9969 &1.0000&1.0000&0.9971&1.0000&0.9978&0.9992\\
  half-parabola & 1.0000 & 0.9818 & 0.9969 &1.0000&0.9823&0.9980&1.0000&0.9980&0.9983\\
  parabola &1.0000 & 0.5070 & 0.9792&1.0000&0.4930&0.9856&1.0000&0.9916&0.9941\\
  exponential &1.0000 &0.9768 & 0.9770&1.0000&0.9770&0.9940&1.0000&0.9940&0.9941\\
  sinusoidal &1.0000 & 0.1289 & 0.8705&1.0000&0.1359&0.9235&1.0000&0.9625&0.9783\\
  sff&1.0000&0.1799&0.8299&1.0000&0.1277&0.8544&1.0000&0.9456&0.9658\\
  snff&1.0000&0.1259&0.8700&1.0000&0.1056&0.9189&1.0000&0.9588&0.9655\\
  svf&1.0000&0.1734&0.8470&1.0000&0.1752&0.8966&1.0000&0.9471&0.9625\\
  circle&0.6770&0.1627&0.4732&0.7098&0.1590&0.4903&0.7098&0.5000&0.5000\\
  normal uncorrelated &0.1320 &0.0859 &0.0070&0.1012&0.0360&0.0030&0.0810&0.0030&0.002\\
  uniform & 0.1134 & 0.0567 & 0.0530&0.1012&0.0480&0.0480&0.0800&0.0400&0.0340 \\

  \hline
\end{tabular}
\end{table*}

\begin{table*}
\renewcommand{\arraystretch}{1.3}
\caption{Values of dependence measure for bivariate normal  distribution with a given correlation coefficient for data sets of varying sizes}
\label{normDat}
\begin{tabular}{|c||c|c|c||c|c|c||c|c||c|}
  \hline
  Number of points&\multicolumn{3}{|c||}{1000}&\multicolumn{3}{|c||}{2000}&\multicolumn{2}{|c||}{5000}&10000\\
  \hline
  cor $(\rho)$ & MINE &Dcor & proposed & MINE &Dcor & proposed&MINE&proposed&proposed \\
  \hline
  1 & 1.0000 &1.0000& 0.8007&1.0000&1.0000 &0.8181&1.0000&0.8176&0.8211\\
  0.95 & 0.8012 & 0.9326 & 0.6701& 0.7705&  0.9321&0.6727& 0.7528& 0.6775&0.6624 \\
  0.9 &0.6249& 0.8580&0.5400& 0.6237&0.8538& 0.5517&0.6338& 0.5593& 0.5431\\
  0.8 &0.5280& 0.7566&0.3721& 0.4892 &0.7449& 0.3708& 0.4756& 0.3778& 0.3597\\
  0.7 &0.4210& 0.6430&0.2545& 0.3774 &0.6475& 0.2801& 0.3529& 0.3008 & 0.2683\\
  0.6&0.3218&  0.5700&0.2165& 0.2671 &0.5235& 0.2014& 0.2081& 0.2005 &0.1774\\
  0.5&0.2707&  0.4642&0.1630& 0.2245 &0.4467& 0.1568& 0.2162& 0.1192 &0.1266\\
  0.3&0.1744&  0.2511&0.1377& 0.1502 &0.2547& 0.0446& 0.1181& 0.0475 &0.0406\\
  0.01&0.1400&   0.0600&0.0084& 0.1000 &0.0300& 0.0047& 0.0700& 0.0017 &0.0008\\

  \hline
\end{tabular}
\end{table*}

\begin{table*}
\renewcommand{\arraystretch}{1.3}
\caption{MIDI, MINE and dcor results for different noise levels using 1000 data points}
\label{MIDInoisy}
\begin{tabular}{|c|c|c|c|c|c|c|c|c|}

  \hline
  Function/Noise level ($\sigma^2$) &Method& $10^{-6}$ &$10^{-4}$& $10^{-2}$ &$10^{-1}$&$1$&$10$&$10^2$\\
  \hline
  \multirow{3}{*}{line} &MIDI &0.9983 &0.9701 &0.7825&0.3356&0.1496&0.1293&0.0901\\
                        &MINE &1.0000 &1.0000 &0.8752&0.4173&0.1831&0.1382&0.1326\\
                        &dcor &1.0000 &1.0000 &0.9383 &0.6592&0.2339&0.0654&0.0415\\
  \hline
  \multirow{3}{*}{half-parabola} &MIDI& 0.9854 &0.9588 &0.7106&0.2912&0.1309&0.1154&0.0805\\
                                 &MINE   & 1.0000 &1.0000 &0.8001&0.4103&0.1624&0.1159&0.1413\\
                                 &dcor   & 0.9821 &0.9804 &0.9236&0.6548&0.2292&0.1457&0.0502\\
  \hline
  \multirow{3}{*}{parabola} &MIDI &0.9737&0.89364 &0.2345&0.1589&0.1212&0.1254&0.0789\\
                            &MINE &1.0000&0.9763 &0.3301&0.1836&0.1412&0.1283&0.1367\\
                            &dcor &0.4978&0.4899 &0.2879&0.1301&0.0562&0.0432&0.0512\\
  \hline
  \multirow{3}{*}{ exponential} &MIDI &0.9944 &0.9856 &0.9582&0.8927&0.6339&0.2691&0.1134\\
                                &MINE &1.0000 &1.0000 &1.0000&0.9632&0.7139&0.3576&0.1548\\
                                &dcor &0.9773 &0.9773 &0.9771&0.9678&0.8832&0.5626&0.2309\\
 \hline
  \multirow{3}{*}{sinusoidal} &MIDI &0.8751 & 0.8745 &0.8367&0.5756&0.2511&0.1165&0.0869\\
                                &MINE &1.0000 & 1.0000 &0.9756&0.7839&0.348&0.1564&0.1385\\
                                &dcor &0.2595 & 0.2593 &0.2556 &0.2143&0.1584&0.0987&0.0412\\
 \hline
  \multirow{3}{*}{ sinusoidal (fourier frequency)}&MIDI &0.8267 &0.8047 &0.7893&0.5745&0.2678&0.1596&0.0783\\
                                &MINE &1.0000 &1.0000 &0.9945&0.7688&0.2889&0.1648&0.1374\\
                                &dcor &0.1596 &0.1538 &0.1512&0.1287&0.0774&0.0732&0.0451\\

 \hline
  \multirow{3}{*}{  sinusoidal (non fourier frequency)} &MIDI&0.8545 &0.8543 &0.8275&0.5478&0.2463&0.1574&0.0914\\
                                                        &MINE &1.0000&1.0000 &0.9732&0.8782&0.2919&0.1587&0.1416\\
                                                        &dcor &0.1282&0.1285 &0.1233&0.12&0.0775&0.0657&0.0548\\
 \hline
  \multirow{3}{*}{  sinusoidal (varying frequency)} &MIDI&0.8314&0.8198 &0.7856 &0.5571&0.2352&0.1267&0.0805\\
                                                    &MINE&1.0000&1.0000 &0.9784&0.7646&0.3023&0.1567&0.1389\\
                                                    &dcor&0.1367&0.1367 &0.1367&0.1354&0.1051&0.0545&0.0512\\
 \hline
  \multirow{3}{*}{  circle} &MIDI&0.4858&0.4817 &0.4784 &0.2778&0.1231&0.0734&0.0654\\
                            &MINE& 0.6567&0.6489 &0.6445 &0.4578&0.1293&0.1423&0.1326\\
                            &dcor& 0.1678&0.1658 &0.1502 &0.1501&0.0560&0.0312&0.0562\\
  \hline
  \multirow{3}{*}{normal uncorrelated} &MIDI&0.0192&0.0121 &0.0127 &0.0249&0.0127&0.0127&0.0072\\
                                                    &MINE&0.1353&0.1245 &0.1332&0.1337&0.1355&0.1257&0.1389\\
                                                    &dcor&0.0515&0.0670 &0.0571&0.0419&0.0573&0.0459&0.0592\\
 \hline
  \multirow{3}{*}{  uniform} &MIDI&0.0551&0.0624 &0.0556 &0.0537&0.0599&0.0637&0.0538\\
                            &MINE& 0.1314&0.1407 &0.1242 &0.1319&0.1293&0.1309&0.1442\\
                            &dcor& 0.0593&0.1026 &0.0593 &0.0524&0.0560&0.0468&0.0562\\
  \hline

\end{tabular}
\end{table*}

\begin{table*}
\renewcommand{\arraystretch}{1.3}
\caption{ MIDI, MINE and dcor results for different noise levels using 2000 data points}
\label{MIDInoisytwo}
\begin{tabular}{|c|c|c|c|c|c|c|c|c|}

  \hline
  Function/Noise level ($\sigma^2$) &Method& $10^{-6}$ &$10^{-4}$& $10^{-2}$ &$10^{-1}$&$1$&$10$&$10^2$\\

 \hline
  \multirow{3}{*}{  line}&MIDI& 0.9984 &0.9738 &0.6956 &0.3365&0.1105&0.1000&0.0798\\
                        &MINE&  1.0000 &1.0000 &0.8202&0.4171&0.1510&0.1101&0.1115\\
                        &dcor & 1.0000 &1.0000 &0.9312 &0.6403&0.2756&0.0982&0.0412\\
 \hline
  \multirow{3}{*}{  half-parabola}&MIDI & 0.9952 &0.9697 &0.6911&0.2932&0.1119&0.1&0.0911\\
                                    &MINE& 1.0000 &1.0000 &0.8216&0.4138&0.1536&0.1201&0.1092\\
                                    &dcor& 0.9820 &0.9857 &0.9292&0.6312&0.2157&0.0984&0.0537\\
 \hline
  \multirow{3}{*}{  parabola}&MIDI &0.9771&0.8936 &0.2434&0.1259&0.0948&0.0884&0.0856\\
                             &MINE &1.0000&0.9609 &0.2889&0.1263&0.1167&0.1183&0.1010\\
                             &dcor &0.4967&0.4839 &0.2736&0.0912&0.0564&0.0455&0.0512\\
 \hline
  \multirow{3}{*}{  exponential}&MIDI &0.9941 &0.9932 &0.9867&0.9092&0.6483&0.2349&0.0918\\
                                &MINE &1.0000 &1.0000 &1.0000&0.9465&0.7278&0.3297&0.1210\\
                                &dcor &0.9778 &0.9787 &0.9773&0.9678&0.8834&0.5664&0.2364\\
 \hline
  \multirow{3}{*}{  sinusoidal} &MIDI &0.9275 & 0.9275 &0.8556&0.5674&0.2438&0.109&0.0756\\
                                &MINE &1.0000 & 1.0000 &0.9778&0.7562&0.2894&0.1196&0.1012\\
                                &dcor &0.2591 & 0.2591 &0.2556 &0.2149&0.1587&0.0923&0.0476\\
   \hline
  \multirow{3}{*}{sinusoidal (fourier frequency)}&MIDI& 0.8945 &0.8876 &0.8168&0.5679&0.2575&0.119&0.07012\\
                                                &MINE&1.0000 &0.9942 &0.9789&0.7689&0.2864&0.1317&0.1089\\
                                                &dcor&0.1184 &0.1182 &0.1125&0.1035&0.0743&0.0527&0.0536\\
  \hline
  \multirow{3}{*}{ sinusoidal (non fourier frequency)}&MIDI &0.8986 &0.8873 &0.8495&0.5636&0.2212&0.1226&0.0896\\
                                                     &MINE &1.0000&0.9912 &0.9873&0.8214&0.22743&0.1261&0.1029\\
                                                     &dcor &0.1156&0.1195 &0.1193&0.1013&0.0776&0.0396&0.0295\\
   \hline
  \multirow{3}{*}{sinusoidal (varying frequency)}&MIDI &0.8976&0.8969 &0.8664 &0.5389&0.2524&0.1256&0.0895\\
                                                &MINE &1.000&1.000 &0.9757&0.7644&0.3032&0.1348&0.1191\\
                                                &dcor &0.1686&0.1656 &0.1539&0.1137&0.10524&0.0518&0.0238\\
 \hline
  \multirow{3}{*}{  circle} &MIDI &0.4947&0.4919 &0.4785 &0.2756&0.1235&0.0759&0.0512\\
                            &MINE&0.6858&0.6814 &0.6448 &0.4239&0.1109&0.1097&0.1116\\
                            &dcor &0.1601&0.1600 &0.1576 &0.1437&0.0609&0.0395&0.0208\\

  \hline
  \multirow{3}{*}{normal uncorrelated} &MIDI&0.0066&0.0058 &0.0050 &0.0068&0.0053&0.0039&0.0059\\
                                                    &MINE&0.1107&0.1010 &0.1084&0.1145&0.1126&0.1257&0.1389\\
                                                    &dcor&0.0515&0.0466 &0.0389&0.0400&0.0356&0.0496&0.0396\\
 \hline
  \multirow{3}{*}{  uniform} &MIDI&0.0470&0.0550 &0.0476 &0.0504&0.0427&0.0476&0.0452\\
                            &MINE& 0.1054&0.1120 &0.1102 &0.1080&0.1048&0.1156&0.1060\\
                            &dcor& 0.0480&0.0416 &0.0375 &0.0474&0.0462&0.0325&0.0465\\
  \hline

\end{tabular}
\end{table*}

\begin{table*}
\renewcommand{\arraystretch}{1.3}
\caption{MIDI and MINE results for different noise levels using 5000 data points}
\label{MIDInoisyfive}
\begin{tabular}{|c|c|c|c|c|c|c|c|c|}

  \hline
  Function/Noise level ($\sigma^2$)&Method & $10^{-6}$ &$10^{-4}$& $10^{-2}$ &$10^{-1}$&$1$&$10$&$10^2$\\
 \hline
  \multirow{2}{*}{  line} &MIDI& 0.9985 &0.9643 &0.6969 &0.3372&0.1125&0.0915&0.0773\\
                            &MINE& 1.0000 &0.9983 &0.7865&0.3174&0.1248&0.0913&0.0792\\

 \hline
  \multirow{2}{*}{ half-parabola} &MIDI& 0.9952 &0.9647 &0.6914&0.2996&0.1162&0.0954&0.0615\\
                                    &MINE & 1.0000 &0.9984 &0.8679&0.3136&0.1276&0.0975&0.0798\\
  \hline
  \multirow{2}{*}{ parabola}&MIDI &0.9776&0.8936 &0.2467&0.1229&0.0976&0.0813&0.0698\\
                            &MINE&0.9976&0.9558 &0.2712&0.1543&0.9863&0.0854&0.0813\\

  \hline
  \multirow{2}{*}{ exponential}&MIDI &0.9940 &0.9932 &0.9717&0.9025&0.5814&0.2319&0.0726\\
                                &MINE&1.0000 &1.0000 &1.0000&0.9465&0.6948&0.2865&0.1039\\
  \hline
  \multirow{2}{*}{ sinusoidal}&MIDI &0.9673 & 0.9587 &0.8532&0.5887&0.2394&0.1567&0.0698\\
                                &MINE&1.0000 & 1.0000 &0.9664&0.7785&0.2646&0.1584&0.0816\\
  \hline
  \multirow{2}{*}{ sinusoidal (fourier frequency)}&MIDI&0.9325 &0.9335 &0.8854&0.5432&0.2125&0.1078&0.0697\\
                                                    &MINE&1.0000 &0.9956 &0.9673&0.7856&0.2634&0.1180&0.0798\\
  \hline
  \multirow{2}{*}{ sinusoidal (non fourier frequency)}&MIDI&0.9513 &0.9346&0.8815&0.5558&0.2185&0.1031&0.0714\\
                                                        &MINE&1.0000&0.9931 &0.9652&0.7714&0.2679&0.1068&0.0793\\
 \hline
  \multirow{2}{*}{  sinusoidal (varying frequency)} &MIDI&0.9532&0.9502&0.8843&0.5409&0.2167&0.1006&0.0712\\
                                                    &MINE&1.0000&0.9980&0.9574&0.7828&0.2674&0.0959&0.0812\\
  \hline
  \multirow{2}{*}{ circle}&MIDI&0.4959&0.4814 &0.4447 &0.2219&0.0764&0.0694&0.0396\\
                            &MINE&0.6903&0.6598 &0.6443 &0.4583&0.1209&0.1004&0.0712\\
  \hline

\multirow{3}{*}{normal uncorrelated} &MIDI&0.0029&0.0024 &0.0049 &0.0018&0.0023&0.0040&0.0031\\
                                    &MINE&0.0763&0.0851 &0.0756&0.07424&0.0815&0.0794&0.0767\\
 \hline
  \multirow{3}{*}{  uniform} &MIDI&0.0481&0.0392 &0.0429 &0.0409&0.0466&0.0390&0.0409\\
                            &MINE& 0.0775&0.0810 &0.0815 &0.0801&0.0823&0.0792&0.0828\\
  \hline

\end{tabular}
\end{table*}

\begin{table*}
\renewcommand{\arraystretch}{1.3}
\caption{MIDI Results for different noise levels using 10000 data points}
\label{MIDInoisyten}
\begin{tabular}{|c|c|c|c|c|c|c|c|}

  \hline
  Function/Noise level ($\sigma^2$) & $10^{-6}$ &$10^{-4}$& $10^{-2}$ &$10^{-1}$&$1$&$10$&$10^2$\\
  \hline
  line & 0.9987 &0.9656 &0.6694 &0.2432&0.1123&0.0678&0.0627\\
  half-parabola & 0.9952 &0.9646 &0.6914&0.2874&0.1124&0.0858&0.0653\\
  parabola &0.9864&0.8654 &0.2485&0.1506&0.0987&0.0795&0.0697\\
  exponential &0.9971 &0.9954 &0.9667&0.8823&0.5695&0.2454&0.0690\\
  sinusoidal &0.9812 & 0.9736 &0.8557&0.5896&0.2314&0.0926&0.0798\\
  sinusoidal (fourier frequency)&0.9735 &0.9648 &0.8426&0.5343&0.2218&0.1109&0.0606\\
  sinusoidal (non fourier frequency)&0.9767 &0.9689&0.8529&0.5446&0.2297&0.0818&0.0704\\
  sinusoidal (varying frequency)&0.9704&0.9603 &0.8534 &0.5268&0.2217&0.0917&0.0742\\
  circle&0.4965&0.4812 &0.4276 &0.1737&0.0547&0.0549&0.0587\\
normal uncorrelated &0.0018&0.0016 &0.0013 &0.0019&0.0011&0.0007&0.0009\\
uniform&0.0396&0.0392 &0.0386 &0.0333&0.0468&0.0320&0.0389\\

  \hline
\end{tabular}
\end{table*}

\begin{table*}
\renewcommand{\arraystretch}{1.3}
\caption{MIDI, MINE and dcor results for different noise levels using 1000 data points for bivariate normal distribution}
\label{MIDInoisynorm}
\begin{tabular}{|c|c|c|c|c|c|c|c|c|}

  \hline
  Correlation coeff/Noise level ($\sigma^2$) &Method& $10^{-6}$ &$10^{-4}$& $10^{-2}$ &$10^{-1}$&$1$&$10$&$10^2$\\

 \hline
  \multirow{3}{*}{  1 }&MIDI  &0.8121 &0.8101&0.7989&0.6445&0.2824&0.0464&0.0157\\
                        &MINE  &1.0000    &1.0000  &0.9765   &0.8145    &0.4132   &0.1855   &0.1178\\
                        &dcor &1.0000    &0.9963  &0.9941   &0.9317    &0.6678   &0.2737   &0.0412\\
 \hline
  \multirow{3}{*}{  0.95}&MIDI& 0.6865 &0.6558 &0.6501&0.5282&0.2714&0.0155&0.0123\\
                        &MINE&0.7617 &0.7608 &0.7212 &0.6505   &0.3816   &0.1516   &0.1367\\
                        &dcor&0.9114 &0.9215 &0.9109 &0.86576   &0.6312   &0.2201   &0.0794\\
 \hline
  \multirow{3}{*}{  0.9}&MIDI &0.5212&0.5354 &0.5263&0.4674&0.2412&0.0237&0.0071\\
                        &MINE&0.6814 &0.6505 &0.6126 &0.5637   &0.3509   &0.1796   &0.1449\\
                        &dcor&0.8664 &0.8657 &0.8419 &0.8013   &0.5943   &0.2516   &0.0659\\
 \hline
  \multirow{3}{*}{  0.8}&MIDI &0.3764 &0.3679 &0.3243&0.3278&0.2091&0.0236&0.0102\\
                        &MINE&0.5001  &0.5101 &0.5297   &0.4712   &0.3219   &0.1311 &0.1453\\
                        &dcor&0.7578  &0.7417 &0.7532   &0.7267   &0.5323   &0.1607 &0.0516\\
 \hline
  \multirow{3}{*}{  0.7}&MIDI & 0.2601 &0.2600&0.2597&0.2254&0.1096&0.0305&0.0077\\
                        &MINE&0.3627&0.4679  &0.4234  &0.3684   &0.2882   &0.1498   &0.1205\\
                        &dcor&0.6225&0.6431  &0.6439  &0.6295   &0.4842   &0.2010   &0.0496\\
 \hline
  \multirow{3}{*}{  0.6}&MIDI &0.1912 &0.1911 &0.1853&0.1149&0.0815&0.0256&0.0050\\
                        &MINE&0.3549 &0.3192 &0.3189&0.2855&0.2365&0.1401&0.1384\\
                        &dcor&0.5762 &0.5384 &0.5559&0.5214&0.4184&0.1677&0.0742\\

 \hline
  \multirow{3}{*}{  0.5}&MIDI &0.1198 &0.1203 &0.1219&0.1163&0.0738&0.0237&0.0094\\
                        &MINE &0.2643 &0.2648 &0.2705&0.2375&0.1839&0.1596&0.1263\\
                        &dcor &0.4677 &0.4674 &0.4701&0.4183&0.2915&0.1584&0.0635\\
 \hline
  \multirow{3}{*}{  0.3}&MIDI &0.0578 &0.0594 &0.0527&0.0472&0.0268&0.0100&0.0078\\
                        &MINE& 0.1745 &0.1593 &0.1724&0.1674&0.1465&0.1315&0.1300\\
                        &dcor&0.2521 &0.2520 &0.2519&0.2147&0.1987&0.0925&0.0456\\
 \hline
  \multirow{3}{*}{  0.01 }&MIDI  &0.0035&0.0031 &0.0029 &0.0032&0.0030&0.0027&0.0031\\
                        &MINE&0.1312&0.1285 &0.1314 &0.1358&0.1349&0.1237&0.1372\\
                        &dcor&0.0503&0.0514 &0.0495 &0.0496&0.0435&0.0503&0.0459\\

  \hline
\end{tabular}
\end{table*}

\begin{table*}
\renewcommand{\arraystretch}{1.3}
\caption{MIDI, MINE and dcor results for different noise levels using 2000 data points for bivariate normal data}
\label{MIDInoisytwonorm}
\begin{tabular}{|c|c|c|c|c|c|c|c|c|}

  \hline
  Correlation coeff/Noise level ($\sigma^2$) &Method& $10^{-6}$ &$10^{-4}$& $10^{-2}$ &$10^{-1}$&$1$&$10$&$10^2$\\
 \hline
  \multirow{3}{*}{  1 }&MIDI  &0.8212 &0.8208&0.8203&0.6676&0.2512&0.0534&0.01234\\
                        &MINE&1    &1.0000  &0.9634   &0.7745    &0.3528   &0.1465   &0.1167\\
                        &dcor&1    &0.9923  &0.9919   &0.9318    &0.6545   &0.2673   &0.1126\\
 \hline
  \multirow{3}{*}{  0.95}&MIDI& 0.6576 &0.6534 &0.6463&0.5675&0.2495&0.0519&0.0120\\
                         &MINE&0.7656 &0.7736 &0.7459 &0.6648   &0.3319   &0.1573   &0.1132\\
                         &dcor&0.9348 &0.9375 &0.9368 &0.86536   &0.5919   &0.2754   &0.1185\\
 \hline
  \multirow{3}{*}{  0.9}&MIDI &0.5575&0.5633 &0.5631&0.5589&0.2456&0.0473&0.0125\\
                        &MINE&0.6547 &0.6589 &0.6539 &0.5672   &0.3294   &0.1385   &0.0987\\
                        &dcor& 0.8794 &0.86539 &0.8657 &0.8284   &0.5874   &0.2185   &0.1101\\
 \hline
  \multirow{3}{*}{  0.8}&MIDI &0.4228 &0.3956 &0.3323     &0.3318   &0.2049    &0.0467&0.0121\\
                        &MINE&0.4587  &0.4514 &0.4537   &0.4589   &0.2743   &0.1379 &0.1143\\
                        &dcor&0.7587  &0.7416 &0.7534   &0.7267   &0.5329   &0.1656 &0.0712\\
 \hline
  \multirow{3}{*}{  0.7}&MIDI & 0.2651 &0.2639&0.2634&0.2252&0.1185&0.0297&0.0081\\
                        &MINE&0.3818&0.3845  &0.3468  &0.3574   &0.2356   &0.1312   &0.1109\\
                        &dcor&0.6655&0.6587  &0.6552  &0.6376   &0.4748   &0.2105   &0.0814\\
 \hline
  \multirow{3}{*}{  0.6}&MIDI &0.1967 &0.1954 &0.1832&0.1146&0.0867&0.0253&0.0052\\
                        &MINE&0.3284 &0.3192 &0.3563&0.2923&0.2598&0.1248&0.1034\\
                        &dcor&0.5426 &0.5345 &0.5311&0.5301&0.4189&0.2498&0.0794\\
 \hline
  \multirow{3}{*}{  0.5}&MIDI &0.1491 &0.1491 &0.1487&0.1375&0.0767&0.0325&0.0106\\
                        &MINE&0.2113 &0.2254 &0.2168&0.2239&0.1628&0.1155&0.1128\\
                        &dcor&0.4335 &0.4453 &0.4357&0.4402&0.3298&0.1217&0.0654\\
 \hline
  \multirow{3}{*}{  0.3}&MIDI &0.0398 &0.0399 &0.0401&0.0390&0.0212&0.0210&0.0060\\
                        &MINE&0.1512 &0.1679 &0.1516&0.1458&0.1638&0.1031&0.1000\\
                        &dcor&0.27 &0.26&0.28&0.23&0.13&0.07&0.04\\
 \hline
  \multirow{3}{*}{ 0.01}&MIDI  &0.0039&0.0038 &0.0046 &0.0030&0.0031&0.0028&0.0030\\
                        &MINE &0.1193&0.1225 &0.1156 &0.1039&0.1953&0.09849&0.1148\\
                        &dcor &0.0350&0.0425 &0.0378 &0.0335&0.0321&0.0383&0.0409\\

  \hline
\end{tabular}
\end{table*}

\begin{table*}
\renewcommand{\arraystretch}{1.3}
\caption{MIDI Results for different noise levels using 5000 data points for bivariate normal distribution}
\label{MIDInoisyfivenorm}
\begin{tabular}{|c|c|c|c|c|c|c|c|c|}

  \hline
  Correlation coeff/Noise level ($\sigma^2$) &Method& $10^{-6}$ &$10^{-4}$& $10^{-2}$ &$10^{-1}$&$1$&$10$&$10^2$\\  \hline
 \hline
  \multirow{2}{*}{  1 }&MIDI  &0.8271 &0.8123&0.7938&0.6563&0.2774&0.0389&0.0080\\
                        &MINE&1.0000    &1.0000 &0.9532  &0.7649    &0.3375   &0.1227   &0.0998\\
 \hline
  \multirow{2}{*}{  0.95}&MIDI& 0.6884 &0.6867 &0.6724&0.5756&0.2786&0.0316&0.0085\\
                            &MINE&0.7689 &0.7465 &0.7436 &0.6487   &0.3387   &0.1354   &0.0958\\
 \hline
  \multirow{2}{*}{  0.9}&MIDI &0.5368&0.5337 &0.5314&0.4635&0.2227&0.0384&0.0058\\
                        &MINE&0.6302 &0.6548 &0.6467 &0.5698   &0.3287   &0.1387   &0.0879\\
 \hline
  \multirow{2}{*}{  0.8}&MIDI &0.3898 &0.3895 &0.3886&0.3556&0.1724&0.0345&0.0054\\
                        &MINE &0.4517  &0.4398 &0.4412   &0.4501   &0.2711   &0.1267 &0.0795\\
 \hline
  \multirow{2}{*}{  0.7} &MIDI& 0.2698 &0.2468&0.2545&0.2254&0.1287&0.0163&0.0021\\
                        &MINE&0.3513&0.3545  &0.3468  &0.3558   &0.2349   &0.1382   &0.1159\\
 \hline
  \multirow{2}{*}{  0.6} &MIDI&0.1973 &0.1971 &0.1868&0.1198&0.0887&0.0257&0.0052\\
                            &MINE&0.2748 &0.2676 &0.2769&0.2696&0.1521&0.1098&0.0812\\
 \hline
  \multirow{2}{*}{  0.5} &MIDI&0.1239 &0.1228 &0.1217&0.1137&0.0756&0.0162&0.0027\\
                        &MINE &0.2167 &0.2107 &0.2187&0.2200&0.1359&0.0936&0.0978\\
 \hline
  \multirow{2}{*}{  0.3} &MIDI&0.0489 &0.0433 &0.0501&0.0415&0.0237&0.0201&0.0021\\
                        &MINE&0.1219 &0.1143 &0.1321&0.1129&0.1119&0.0873&0.0855\\
 \hline
  \multirow{2}{*}{  0.01 }  &MIDI&0.0036&0.0035 &0.0024 &0.0032&0.0031&0.0028&0.0034\\
                            &MINE&0.0801&0.0932 &0.0897 &0.0947&0.1023&0.0963&0.0885\\

  \hline
\end{tabular}
\end{table*}

\begin{table*}
\renewcommand{\arraystretch}{1.3}
\caption{MIDI Results for different noise levels using 10000 data points for bivariate normal data}
\label{MIDInoisytennorm}
\begin{tabular}{|c|c|c|c|c|c|c|c|}

  \hline
  Correlation coeff/Noise level ($\sigma^2$) & $10^{-6}$ &$10^{-4}$& $10^{-2}$ &$10^{-1}$&$1$&$10$&$10^2$\\  \hline
  1   &0.8224 &0.8197&0.8175&0.6569&0.2674&0.0488&0.0021\\
  0.95& 0.6671 &0.6132 &0.6121&0.5574&0.2269&0.0459&0.0046\\
  0.9 &0.5132&0.4878 &0.4869&0.4386&0.1746&0.0473&0.0041\\
  0.8 &0.3789 &0.3673 &0.3698&0.3269&0.1677&0.0238&0.0051\\
  0.7 & 0.2586 &0.2485&0.2501&0.2253&0.1187&0.0256&0.0031\\
  0.6 &0.1290 &0.1197 &0.1184&0.1196&0.0672&0.01267&0.0029\\
  0.5 &0.1288 &0.1246 &0.1255&0.1139&0.0727&0.0168&0.0026\\
  0.3 &0.0396 &0.0386 &0.0386&0.0400&0.0284&0.0058&0.0020\\
  0.01   &0.0035&0.0029 &0.0014 &0.0032&0.0015&0.0026&0.0012\\

  \hline
\end{tabular}
\end{table*}

\begin{table*}
\renewcommand{\arraystretch}{1.3}
\caption{Different data sets used determining power. Values of standard deviation used for each function is given by product of corresponding noise scale with $1/10, 2/10, \cdots ,30/10 $}
\label{PowerFuncs}
\begin{tabular}{|c|c|c|c|}

  \hline
  Function & Description&Domain&Noise scale\\
  \hline
  line & $y=x$& [0,1]&1 \\
  quadratic & $y=4*(x-.5)^2$ & [0,1]&1 \\
  cubic & $y=128*(x-1/3)^3-48*(x-1/3)^2-12*(x-1/3)$ & [0,1]&10 \\
  sin 1/8 &$y=\sin(4*pi*x)$&[0,1]&2\\
  sin 1/2 &$y=\sin(16*pi*x)$ & [0,1]&1\\
  fourth root&$y=x^{1/4}$&[0,1]&1\\
  circle&\parbox{7cm}{$y=(2*z-1) * (\sqrt{1 - (2*x - 1)^2})$, where $z$ is randomly chosen from $\{0,1\}$ } &[0,1]&1/4\\
  step &$y = 1 \text{,if }(x > 0.5)\text{, }y=0\text{,otherwise}$&[0,1]&5\\

  \hline
\end{tabular}
\end{table*}

\begin{table*}
\renewcommand{\arraystretch}{1.3}
\caption{Analysis of Yeast genes}
\label{realdata}
\begin{tabular}{|c|c|c|c|c|c|}

  \hline
  Number & Name of genes&Genes identifier&MINE&dcor&MIDI\\
  \hline
1&NUP170&YBL079W&0.2835&0.2991&0.8235\\
2&PCH2&YBR186W&0.2160&0.2581&0.8250\\
3&PERT18&YCR020C&0.2614&0.2737&0.8230\\
4&LUC7&YDL087C&0.2588&0.2695&0.8150\\
5&YDL228C&YDL228C&0.2588&0.2914&1.0000\\
6&YDL238C&YDL238C&0.2588&0.2649&0.8333\\
7&RAD61&YDR014W&0.1742&0.2562&0.8050\\
8&APC4&YDR118W&0.2835&0.2935&0.8333\\
9&STE14&YDR410C&0.2800&0.2433&0.8141\\
10&SAM2&YDR502C&0.2858&0.2774&0.8351\\
11&MMS21&YEL019C&0.2547&0.2648&0.8141\\
12&HAT2&YEL056W&0.2558&0.2951&0.8050\\
13&SPR6&YER115C&0.2445&0.2973&0.8369\\
14&FAU1&YER183C&0.2858&0.2929&0.8050\\
15&CAF16&YFL028C&0.2614&0.2942&0.8073\\
16&YFR041C&YFR041C&0.2983&0.2737&0.9344\\
17&TOS3&YGL179C&0.2866&0.2392&0.8318\\
18&YGR251W&YGR251W&0.2614&0.2534&0.8333\\
19&YHR045W&YHR045W&0.2230&0.2457&0.8318\\
20&YIL055C&YIL055C&0.2473&0.2913&0.9047\\
21&PEX1&YKL197C&0.2665&0.2964&0.8923\\
22&RRN5&YLR141W&0.2149&0.2777&0.8230\\
23&YLR152C&YLR152C&0.2121&0.2981&0.8369\\
24&YLR287C&YLR287C&0.2762&0.2891&0.8032\\
25&SNZ1&YMR096W&0.2762&0.2762&0.8032\\
26&YNL324W&YNL324W&0.2762&0.2422&0.8318\\
27&YOL138C&YOL138C&0.2614&0.2697&0.8318\\
28&OST3&YOR085W&0.2800&0.2885&0.8834\\
29&ELG1&YOR144C&0.2753&0.2701&0.8333\\
30&SPP2&YOR148C&0.2762&0.2613&0.8541\\
31&SVL3&YPL032C&0.2614&0.2382&0.8150\\
32&SSN3&YPL042C&0.2160&0.2495&0.9380\\
33&CUP9&YPL177C&0.2365&0.2802&0.8318\\
34&GIT1&YCR098C&0.9986&0.9388&0.9424\\
35&CPR6&YLR216C&0.9986&0.8917&0.9358\\
36&HSP12&YFL014W&0.8144&0.4845&0.8848\\

  \hline
\end{tabular}
\end{table*}

\FloatBarrier
\begin{figure*}[h!]
\caption{Power of different indices of dependence}
\label{MIDIpower}
\includegraphics[width=\textwidth]{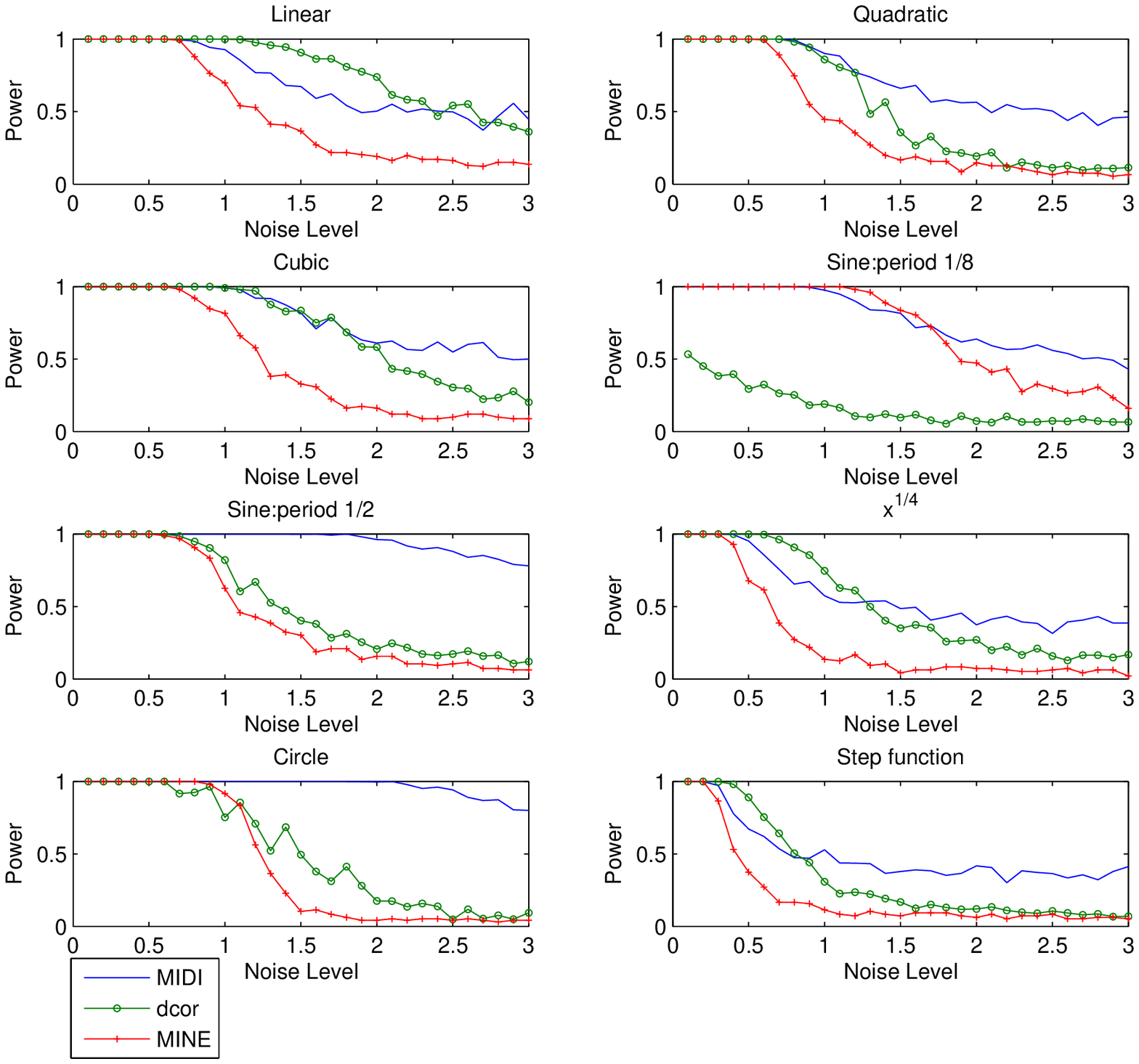}
\end{figure*}

\FloatBarrier


\bibliographystyle{spbasic}      
\bibliography{MIDIbib}   
\bigskip
\appendix
\noindent{\bf Appendix}

\section{Theorems and Proofs}
\label{thpf}

\subsection{$L-1 consistency$ of density estimates}
Let ${\mathbf{R}}^d$ denote d-dimensional Euclidean space. An ordered sequence $x_1, . . . , x_n \in {\mathbf{R}}^d$ will be denoted by $x_1^n$. By a partition of ${\mathbf{R}}^d$ we mean a finite collection $\pi = \{A_1, \cdots A_r\}$ of Borel-measurable subsets of ${\mathbf{R}}^d$, referred to as cells, with the property that (i) $\bigcup^r_{j=1}A_j ={\mathbf{R}}^d$ and (ii) $A_i \cap A_j = \emptyset $, if $i \not= j$. Let $|\pi|$ denote the number of cells in $\pi$.
Let $\mathcal{A}$ be a (possibly infinite) family of partitions of ${\mathbf{R}}^d$. The maximal cell count of $\mathcal{A}$ is given by
$m(\mathcal{A}) = \sup_{\pi \in \mathcal{A}} |\pi|$.
The complexity of $\mathcal{A}$ will be measured by a combinatorial quantity similar to the growth function for classes of sets that was proposed by \cite{VCpaper}. Fix $n$ points $x_1,\cdots , x_n \in {\mathbf{R}}^d$  and let $B = \{x_1,\cdots , x_n\}$. Let $\Delta(\mathcal{A}, x_1^n)$ be the number of distinct partitions $\{A_1 \cap B, \cdots ,A_r \cap B\}$ of the finite set $B$ that are induced by a partition $\{A_1,\cdots ,A_r\} \in \cal{A}$. It is easy to see that $\Delta(\mathcal{A}, x_1^n) \leq  m({\mathcal{A}})^n$. The growth function of $\mathcal{A}$ is defined as $\Delta_n^\ast(\mathcal{A}) = \max_{x_1^n \in \mathbf{R}^{n.d} } \Delta(\mathcal{A},x_1^n)$ is the largest number of distinct partitions of any $n$ point subset of ${\mathbf{R}}^d$  that can be induced by the partitions in $\mathcal{A}$. In other words, it is the maximum number of ways in which any set of n fixed points can be partitioned.

The density estimate is produced in two stages from a training set $T_n$ that consists
of $n$ i.i.d. random variables $Z_1,\cdots,Z_n$ taking values in a set $\mathcal{X}={\mathbf{R}}^d$. Using $T_n$ a partition $\pi_n = \pi_n(Z_1,\cdots,Z_n)$ is produced according to a prescribed rule. The partition $\pi_n$ is then used in conjunction with $T_n$ to produce a density estimate. An n-sample partitioning rule for $\mathbf{R}^d$ is a function $\pi_n$ that associates every n-tuple $(z_1,\cdots z_n) \in \mathcal{X}^n$ with a measurable partition of $\mathbf{R}^d$. Applying the rule $\pi_n$ to $Z_1,\cdots ,Z_n$ produces a random partition $\pi_n(Z_1, \cdots ,Z_n)$. A partitioning scheme for $\mathbf{R}^d$ is a sequence of partitioning rules $\Pi = {\pi_1, \pi_2, \cdots}$. Associated with every rule $\pi_n$ there is a fixed, non-random family of partitions $\mathcal{A}_n = {\pi_n(Z_1, . . . , Z_n) : Z_1, . . . , Z_n \in \mathcal{X}}$. Thus every partitioning scheme $\Pi$ is associated with a sequence ${\mathcal{A}_1,\mathcal{A}_2, \cdots}$ of partition families. In what follows the random partitions $\pi_n(Z_1, . . . ,Z_n)$ will be denoted simply by $\pi_n$. With this convention in mind, for every $x \in \mathbf{R}^d$ let $\pi_n[x]$ be the unique cell of $\pi_n$ that contains the point $x$. Let $A$ be any subset of $\mathbf{R}^d$. The diameter of $A$ is the maximum Euclidean distance between two points of $A$, $diam(A) = \sup_{x,y\in A} ||x - y||$.
Let $\mu$ be a probability measure on $\mathbf{R}^d$ having density $f$, so that $\mu(A)=\int_A f(x)dx$ for every Borel subset A of $\mathbf{R}^d$. Let $X_1,X_2, . . .$ be i.i.d. random vectors in $\mathbf{R}^d$, each distributed according to $\mu$, and let $\mu_n$ be the empirical distribution of $X_1, \cdots ,X_n$. Fix a partitioning scheme $\Pi = {\pi_1, \pi_2, \cdots}$ for $\mathbf{R}^d$. Applying the $n^{th}$ rule in $\Pi$ to $X_1, \cdots ,X_n$ produces a partition $\pi_n = \pi_n(X^n_1)$ of $\mathbf{R}^d$. The partition $\pi_n$, in turn, gives rise to a natural histogram estimate of $f$ as follows. For each vector $x \in \mathbf{R}^d$ let
\begin{align}
f_n(x) &= \mu_n(\pi_n[x])/\lambda(\pi_n[x]) \textit{ if } \lambda(\pi_n[x]) < \infty\\
       &=0 \textit{ otherwise}
\end{align}
Here $\lambda$ denotes the Lebesgue measure on $\mathbf{R}^d$. Note that $f_n$ is itself a function of the
training set $X_1,\cdots X_n$, and that $f_n$ is piecewise constant on the cells of $\pi_n$. The sequence
of estimates ${f_n}$ is said to be $strongly$ $L_1-consistent$ if $\int|f(x) - f_n(x)|dx \to 0$ with probability one as $n \to \infty$. A theorem, given by \cite{Lugosi} is stated below using the notations mentioned above.
\begin{theorem}
\label{L1cons}
Let $X1,X2,\cdots$ be i.i.d. random vectors in ${\mathbf{R}}^d$ whose common distribution $\mu$ has a density $f$. Let $ \Pi = \{\pi_1, \pi_2, \cdots\}$ be a fixed partitioning scheme for ${\mathbf{R}}^d$, and let $\mathcal{A}_n$ be the collection of partitions associated with the rule $\pi_n$. As $n$ tends to infinity,
\begin{align}
n^{-1}m(A_n) \to 0\\
n^{-1} \log \Delta_n^\ast (A_n) \to 0\\
\mu\{x : diam(\pi_n[x]) > \gamma \} \to 0
\end{align}
with probability one for every $\gamma >0$, then the density estimates $f_{n}$ are strongly consistent in L1:
\begin{align}
\int|f(x)-f_{n}(x)|dx \to 0
\end{align}
,with probability one.
\qed
\end{theorem}

\subsection{Shrinking cells}
Let $\{U_k\}$, $1\le k\le N$ to  be uniform i.i,d, in [0, 1], ordered to give $\{U_{N,k}^{\ast}\}$ We write $U_{N,0}^{\ast} \equiv 0$ , $U_{N,N+1}^{\ast} \equiv 1$. Then $\gamma_j = U_{N,j+1}^{\ast}-U_{N,j}^{\ast} $ for $0<j<_N$ are called uniform spacings. Let $M_n = \max \gamma_j$ be the largest uniform spacing. Paul Levy gave a heuristic argument to show that
\begin{align}
\lim Pr\{M_n<(\log n+a)/n\}=exp(-exp(-a))\\
\text{,where } -\infty<a<\infty
\end{align}
Further \cite{Darling} proved following theorem
\begin{theorem}
\label{Darlingthm}
Almost surely as $n \to \infty$,
\begin{align}
| M_n-\log(n)/n | = O(\log (\log (n))/n)
\end{align}
It has been specified that this implies for a constant C as $n \to \infty$,
\begin{align}
| M_n-\log (n)/n | &= C(\log (\log (n))/n)
\end{align}
\qed
\end{theorem}

Following theorems establish the upper limit for maximal spacing.
\begin{theorem}
\label{dev1}
Let $X_1,X_2,\cdots$ be an i.i.d. sequence of random variables with a continuous distribution function $F$. Let, $X_{1,n} < X_{2,n} < \cdots <X_{n,n}$ denote the order statistics of $X_{1}, X_{2}, \cdots, X_{n}$, and let $S_i^{(n)}=X_{i+1,n}-X_{i,n}$, $i=1 ,2\cdots,n-1$ define the corresponding spacings. Denote the order statistics of $S_1^{(n)}, \cdots, S_{n-1}^{(n)}$ by $M_{n-1}^{(n)}< \cdots <M_{2}^{(n)}<M_{1}^{(n)}$. If the following conditions are satisfied
\begin{enumerate}
    \item{}$F(x)=P\{X\le x\}$ has a continuous first derivative $f(x)>0$ on $(A,B)$ where $A=inf\{x;F(x)>0\}<B=\{x;F(x)<1\}$.
    \item{}The distribution $F$ has a bounded support $(-\infty<A<B<+\infty)$, and there exists an $x_0 \in (A,B)$ such that, for all $x \in (A,B)$, $x \ne x_0$, $f(x)>f(x_0)>0$.
    \item{}There exists an $r$, $0<r<\infty$, such that
    \begin{align}
    \lim {\inf_{h \downarrow 0} {\frac{f(x_0+h)-f(x_0)}{\lvert h \rvert^r}=d_r}} \text{, }0<d_r<+\infty
    \end{align}
\end{enumerate}
Then for any $p \ge 5$, $k \ge 1$ and $\epsilon>0$,
    \begin{align}
    P(nM_k^{(n)}f(x_0)> \log(n)-(1/r)\log_2(n)\\
    +1/k(2 \log_2(n)+\log_3(n)+\cdots
    +\log_{p-1}(n)\\+(1+\epsilon)\log_p(n))i.o.)=0
    \end{align}
\qed
\end{theorem}
Under the above mentioned conditions \cite{Deheuvels} also states the following

\begin{theorem}
\label{dev2}
In addition to conditions for theorem \ref{dev1} if the following condition is satisfied
\begin{enumerate}
    \item{}There exists an $r$, $0<r<\infty$, such that
    \begin{align}
    \lim {\sup_{h \downarrow 0} {\frac{f(x_0+h)-f(x_0)}{\lvert h \rvert^r}=D_r}} \text{, }0<D_r<+\infty
    \end{align}
\end{enumerate}
and $0<d_r\leq D_r< +\infty$
Then, we have, almost surely,
\begin{align}
\lim{\sup_{n\to\infty}{\frac{nM_k^{(n)}f(x_0)-\log(n)}{\log_2(n)}}}=\frac{2}{k}-\frac{1}{r}
\end{align}
\qed
\end{theorem}

\end{document}